\documentclass[11pt]{article}
\usepackage[margin=1in]{geometry}
\usepackage{hyperref}
\usepackage{amsfonts,amsmath,amsthm,wrapfig}
\usepackage{comment,fullpage}
\usepackage[noend]{algorithmic}
\usepackage[boxed]{algorithm}
\usepackage{xspace}
\usepackage{xcolor}
\usepackage{verbatim}
\usepackage{tikz}
\usetikzlibrary{arrows}

\usepackage{enumerate}
\usepackage{tcolorbox}

\newtheorem{theorem}{Theorem}[section]
\newtheorem{lemma}[theorem]{Lemma}

\newtheorem{corollary}[theorem]{Corollary}

\newcommand\E{\mathbb{E}}
\newcommand\R{\mathbb{R}}

\newcommand\disc{disc}

\newcommand{\set}[1]{\left\{ #1 \right\}}
\newcommand{\ip}[2]{\langle #1, #2 \rangle}

\newcounter{note}[section]

\makeatletter
\newcommand{\labeltext}[2]{%
  \@bsphack
  \csname phantomsection\endcsname 
  \def\@currentlabel{#1}{\label{#2}}%
  \@esphack
}
\makeatother

\date{\vspace{-5ex}}
\begin{document}
\title{The Gram-Schmidt Walk: \\ A Cure for the Banaszczyk Blues }

\author{Nikhil Bansal
\thanks{Department of Mathematics and Computer Science, Eindhoven University of Technology, Netherlands.
Email:
\href{mailto:n.bansal@tue.nl}{n.bansal@tue.nl}.
Supported by a NWO Vidi grant 639.022.211 and an ERC consolidator grant 617951.}
\and
Daniel Dadush\thanks{Centrum Wiskunde and Informatica, Amsterdam.
Email:
\href{mailto:dadush@cwi.nl}{dadush@cwi.nl}.
Supported by NWO Veni grant 639.071.510.
}
\and
Shashwat Garg\thanks{Department of Mathematics and Computer Science, Eindhoven University of Technology, Netherlands.
Email:
\href{mailto:s.garg@tue.nl}{s.garg@tue.nl}.
Supported by the Netherlands Organisation for Scientific Research (NWO) under project no.~022.005.025.
}
\and
Shachar Lovett\thanks{Department of Computer Science and Engineering, University of California, San Diego.
Email:
\href{mailto:slovett@cs.ucsd.edu}{slovett@cs.ucsd.edu}.
Supported by an NSF CCF award 1614023.
}
}

\maketitle

\begin{abstract}
An important result in discrepancy due to Banaszczyk states that for any set of $n$ vectors in $\mathbb{R}^m$ of $\ell_2$ norm at most $1$ and any convex body $K$ in $\mathbb{R}^m$ of Gaussian measure at least half, there exists a $\pm 1$ combination of these vectors which lies in $5K$. This result implies the best known bounds for several problems in discrepancy. Banaszczyk's proof of this result is non-constructive and a major open problem has been to give an efficient algorithm to find such a $\pm 1$ combination of the vectors.

In this paper, we resolve this question and give an efficient randomized algorithm to find a $\pm 1$ combination of the vectors which lies in $cK$ for $c>0$ an absolute constant.
This leads to new efficient algorithms for several problems in discrepancy theory.

\end{abstract}


\allowdisplaybreaks
\section{Introduction}


Let $(X,\mathcal{S})$ be a finite set system, with $X=\{1,2,\dots,n\}$ and $\mathcal{S}=\{S_1,S_2,\dots,S_m\}$ a collection of subsets of $X$. Given a two-coloring $x:X\rightarrow \{-1,1\}$, the discrepancy of a set $S$ is defined as $x(S) = |\sum_{i\in S} x(i)  |$ and measures the imbalance between the number of elements in $S$ colored $-1$ and $1$. The discrepancy of the set system $(X,\mathcal{S})$ is defined as
\[ disc(\mathcal{S}) = \min_{x:X\rightarrow \{-1,1\}} \max_{S\in\mathcal{S}} x(S) \]
and is the minimum imbalance achieved by all the sets in $\mathcal{S}$ over all possible two-colorings.

Discrepancy is defined more generally for any $m\times n$ matrix $A$ as,
\[ disc(A) = \min_{x\in\{- 1,1\}^n} \| Ax\|_{\infty} .\]
That is, minimum achievable $\ell_{\infty}$ norm of the vector $Ax$ over all two-colorings $x$ of the columns of $A$. This can be seen as a vector balancing problem: given vectors $v_1,\dots,v_n \in\mathbb{R}^m$ (specified by the columns of $A$), find a two-coloring $x: [n]\rightarrow \{-1,1\}$ to minimize $\| \sum_{i=1}^n x(i) v_i \|_{\infty}$. The set system view mentioned earlier corresponds to the special case where $A$ is the incidence matrix of the set system with columns indexed by elements in $X$ and rows by sets in $\mathcal{S}$.

Discrepancy is a widely studied topic and has applications to many areas in mathematics and computer science.
In particular in computer science, it arises naturally in computational geometry, data structure lower bounds, rounding in approximation algorithms, combinatorial optimization, communication complexity, pseudorandomness and differential privacy. For more on these connections we refer the reader to \cite{Chazelle,Mat09,Panorama, Sashothesis}.

One of the earliest techniques employed in discrepancy was linear algebraic in nature and similar to the well known iterated rounding technique \cite{B08,BF81,LRSbook}. Though this technique gave surprisingly good bounds for some problems in discrepancy, there remained a big gap between the bounds obtained and the lower bounds known for these problems. 

A huge breakthrough was made in the early 80's with Beck's partial coloring method \cite{Beck81b}, that was further refined by Spencer to the entropy method \cite{Spencer85}. A similar approach based on ideas from convex geometry was developed independently by Gluskin \cite{Glu89}. Roughly speaking, this method guarantees the existence of a coloring of a constant fraction of the elements where every set in the set system incurs a low discrepancy. This is then repeated $O(\log n)$ times in order to get a coloring of all the $n$ elements. The partial coloring method led to improved bounds for many problems in discrepancy and in particular, it led to the famous ``six standard deviations" theorem of Spencer \cite{Spencer85}: given a set system with $n$ sets and $n$ points, there exists a coloring of discrepancy at most $6\sqrt{n}$. This bound is tight up to constant factors.

While the original proofs of the partial coloring method were based on the pigeonhole principle and were non-algorithmic, over the past few years several new algorithmic versions of the partial coloring method have been developed \cite{B10,LM12,Ro14,HSS14,ES14}. In particular, all known applications of partial coloring \cite{Spencer85,Mat09} can now be made algorithmic. These ideas have also led to new results in approximation algorithms \cite{R13,BCKL14,BN15,NTZ13}.

Despite its huge success, the partial coloring method gives sub-optimal bounds for many problems. The reason is that the partial coloring step only colors a constant fraction of the points and hence must be repeated $O(\log n)$ times before all the $n$ points are colored. The discrepancies incurred at each of these steps are independent of each other and in fact can add up adversarially. Consider for instance the long-standing Beck-Fiala problem~\cite{BF81} about discrepancy of low-degree set systems, where every point lies in at most $t$ sets. Here, one round of partial coloring ensures a discrepancy of $O(\sqrt{t})$ to every set, but over all the $O(\log n)$ rounds the discrepancy ends up being $O(\sqrt{t}\log n)$. On the other hand, the Beck-Fiala conjecture is that the discrepancy of such set systems should be $O(\sqrt{t})$. Such logarithmic factor gaps also exist in several other problems in discrepancy for similar reasons.

\paragraph{Banaszczyk's Method.} One of the key results in discrepancy is the following result by Banaszczyk.
\begin{theorem}[\cite{B97}]\label{thm:bana}
Given any convex body $K\subseteq \mathbb{R}^m$ of Gaussian measure $\gamma_m(K)\ge 1/2$, and vectors $v_1,\dots,v_n\in\mathbb{R}^m$ of $\ell_2$ norm at most $1$, there exists a coloring $x:[n]\rightarrow\{-1,1\}$ such that $\sum_{i=1}^n x(i) v_i \in cK$, where $c$ is an absolute constant. In particular, $c=5$ suffices.
\end{theorem}

Here the Gaussian measure $\gamma_m(S)$ of any measurable set $S\subseteq \mathbb{R}^m$ is defined as
\[\gamma_m(S)= \Pr[g \in S]= \int_{y\in S} \frac{1}{(2\pi)^{m/2}}  e^{-\| y\|^2/2}  dy \]
where $g$ is a standard Gaussian random vector in $\mathbb{R}^m$.

In contrast to the partial coloring method, Theorem~\ref{thm:bana} gives a full coloring directly, resulting in improved bounds for several problems in a direct black box way. For instance, Banaszczyk's result implies a discrepancy bound of $O(\sqrt{t\log n})$ for the Beck-Fiala problem and more generally, a $O(\sqrt{\log n})$ bound for the Koml\'{o}s problem. This follows by noting that a $O(\sqrt{\log m})$ scaling of the hypercube $[-1,1]^m$ has Gaussian volume at least half (using standard reductions, we can assume $m\le n$). Similarly,
\cite{MNT14} used it to find coloring with discrepancy at most $O(\sqrt{\log n})$ times the $\gamma_2$-norm, and \cite{Lar14}
used it to give update-query tradeoffs for dynamic data structures.

Theorem \ref{thm:bana} was also used in a very interesting non-black box way in a later work by Banaszczyk \cite{Bana12} to show improved bounds for several variants of the Steinitz problem. Recently, \cite{Sasho16} used this to obtain improved bounds for the Tusnady's problem.
%

Banaszczyk's proof is highly elegant and based on deep ideas from convex geometry. However, his approach is non-algorithmic and finding an algorithmic version of Theorem~\ref{thm:bana} has been a major challenge in discrepancy \cite{R-convex14, Sashothesis, DGLN16}. Partial progress was made on this recently  \cite{BDG16, BG17,harish17} and algorithms achieving the same bounds as Banaszczyk for the Koml\'{o}s problem and for the Steinitz problem in the $\ell_\infty$ norm were obtained. Roughly, these results correspond to the case when the convex body $K$ in Theorem~\ref{thm:bana} is a scaling of the hypercube, and the question about general convex bodies remained open.

In a recent result, \cite{DGLN16} reformulated Banaszczyk's result in terms of
certain subgaussian distributions and reduced the question of finding an
algorithmic version of Banaszczyk's result to constructing such subgaussian
distributions. To state this result, we first need some definitions. A random
vector $Y$ taking values in $\mathbb{R}^m$ is said to be subgaussian with
parameter $\sigma$ (or $\sigma$-subgaussian) if for all $\theta\in \mathbb{R}^{m}$,
\[ \E\left[ e^{\langle \theta,Y\rangle }  \right] \le  e^{(\sigma^2/2)\| \theta \|_2^2} .\]
Observe that a standard Gaussian random vector is 1-subgaussian. We can now state the result in \cite{DGLN16}. Let (P1) denote the following statement\labeltext{(P1)}{stp1}: \\

\begin{minipage}{35em}
\flushleft
``Let $v_1,\ldots,v_n \in \R^m$ be vectors of $\ell_2$ norm at most $1$ and let
$x_0\in[-1,1]^n$. Then there exists a distribution $D$ on $\{-1,1\}^n$,
such that for $x$ sampled from $D$, the random variable $\sum_{i=1}^n (x(i)-x_0(i)) v_i$ is
$\sigma$-subgaussian, for some absolute constant $\sigma>0$.
Moreover, for every $i$ for which $x_0(i) \in \{-1,1\}$ we have $x(i)=x_0(i)$ with probability one."\\
\end{minipage}
\begin{minipage}{5em}
\flushright
(P1)\\
\end{minipage}

\bigskip

\begin{theorem}[\cite{DGLN16}]\label{thm:subgaussian}
Theorem \ref{thm:bana} (up to the exact value of $c$) is equivalent to the statement \ref{stp1}.

More precisely, for $v_1,\dots,v_n \in \R^m$ of $\ell_2$ norm at most $1$ and a convex body $K
\subseteq \R^m$ of Gaussian measure at least $1/2$, there exists $x_0 \in
[-1,1]^n$ such that
\begin{itemize}
\item[(i)] $\sum x_0(i) v_i \in K$,
\item[(ii)] $\Pr_{x \sim D}[\sum_{i=1}^n x(i) v_i \in cK] \ge 1/2$,
\end{itemize}
where $D$ is as in $\ref{stp1}$ and $c := c(\sigma)$ is an absolute constant.

Furthermore, for a symmetric convex body (i.e. $K=-K$), the choice $x_0 = 0$
suffices, and for a general convex body given by a membership oracle, $x_0$ can
be computed in expected polynomial time.
\end{theorem}

The above theorem implies that to get a constructive version of
Theorem~\ref{thm:bana} for any convex body it suffices to give an algorithm that
can efficiently sample from the $\sigma$-subgaussian distributions
in~\ref{stp1}. Furthermore, restricting the sampler to the choice $x_0 = 0$,
gives a \emph{universal algorithm} for finding colorings landing inside any
\emph{symmetric convex body} of Gaussian measure at least $1/2$.

\subsection{Main Result}
In this paper we give the first efficient algorithm for obtaining Banaszczyk's
result (Theorem~\ref{thm:bana}), which also yields a new constructive proof, by
providing a new random walk procedure to sample from the requisite subgaussian
coloring distributions. We dub this procedure the Gram-Schmidt walk and
state its guarantees below:

\begin{theorem}[Gram-Schmidt Walk]
\label{thm:main}
There is a polynomial-time randomized algorithm that, given as input vectors $v_1,\dots,v_n\in\mathbb{R}^m$ of $\ell_2$ norm at most $1$ and $x_0\in [-1,1]^n$, outputs a coloring $x\in\{-1,1\}^n$ such that
the random variable $\sum_{i=1}^n (x(i)-x_0(i))v_i$ is subgaussian with
parameter $\sigma = \sqrt{40} \approx 6.32$.
\end{theorem}



Our algorithm in fact runs in time of $O(n \cdot (n+m)^{\omega})$, where
$\omega$ is the exponent of matrix multiplication.  In particular, it runs in
$n$ iterations (one variable gets colored at each iteration), and in each iteration
the most expensive step is solving a linear system in $n$ variables and $m$
equations, which can be done in time $O((n+m)^{\omega})$.



Comparing to the previous Banaszczyk inspired random walk
samplers, \cite{BDG16,BG17} guaranteed that for each
$i\in[m]$, $\langle e_i, \sum_{i=1}^n (x(i)-x_0(i))v_i \rangle$ is
$O(1)$-subgaussian, where $e_1,\dots,e_m$ are the standard basis vectors in
$\mathbb{R}^m$, and~\cite{DGLN16} guaranteed full $O(\sqrt{\log
n})$-subgaussianity (i.e.~along all directions). That is, prior work either
achieved $O(1)$-subgaussanity only for coordinate directions or full
$O(\sqrt{\log n})$-subgaussianity, whereas Theorem~\ref{thm:main} gives the
first full $O(1)$-subgaussianity guarantee.

Interestingly, for the previous samplers mentioned above, each walk step is
indexed by the solution to a semidefinite program, whereas in contrast, the walk
steps for the new sampler require only basic linear algebra, namely, the Gram-Schmidt orthogonalization.
We note however that the idea for the walk was inspired by the constructive
proof of~\cite{DGLN16} for the existence of solutions to the Koml{\'o}s vector
coloring program of Nikolov~\cite{N13}, where the Gram-Schmidt orthogonalization
plays a crucial role in the analysis.

\subsection{Other Results}
\label{sec:introapp}
Theorem~\ref{thm:main} directly gives an algorithm for previous black box applications of Banaszczyk's result.


\paragraph{Koml\'{o}s problem.} This is a generalization of the Beck-Fiala problem and is defined as follows: given an $m\times n$ matrix $A$ with columns of $\ell_2$ norm at most one, find a coloring $x\in\{-1,1\}^n$ to minimize $disc(A)=\| Ax\|_\infty$.
The Koml\'{o}s conjecture \cite{spencerbook} states that $disc(A)=O(1)$ and is a generalisation of the Beck-Fiala conjecture.
Theorem~\ref{thm:bana} directly gives $O(\sqrt{\log n})$ bound for the Koml\'{o}s problem \cite{B97}. While algorithms to find such a coloring were recently given in \cite{BDG16,BG17,harish17}, theorem~\ref{thm:main} gives another, more direct, algorithm to find such a coloring.
%
%
%

\paragraph{$\ell_p$ Discrepancy.} For $p\in[1,\infty)$, the $\ell_p$ discrepancy of an $m\times n$ matrix $A$ under a coloring $x$ is defined as $\left( \frac{1}{m} \| Ax\|_p^p \right)^{1/p}.$
Matousek~\cite{matouseklp} showed an $\ell_p$ discrepancy bound of $O(p t^{1/2})$ for the Beck-Fiala problem, using partial coloring methods. An improved bound of $O(p^{1/2})$ for the more general Koml\'os
setting follows directly from Banaszczyk's result (and a standard estimate of the Gaussian measure of $\ell_p$-ball). Our result gives an algorithmic version of this bound.
\begin{corollary}
\label{thm:lpdiscrepancy}
There is an efficient randomized algorithm which
given a matrix $A$ with $n$ columns of $\ell_2$ norm at most $1$, and $p\in [1,\infty)$, finds a $\{-1,1\}^n$  coloring
with expected $\ell_p$ discrepancy $O(\sqrt{p})$.
\end{corollary}

Interestingly, a single algorithm given by Theorem~\ref{thm:main} produces a (random) coloring that simultaneously achieves this bound for every $p\in[1,\infty)$.
In contrast, the coloring in Banaszczyk's approach depends on the body $K$ which is different for different values of $p$.


\paragraph{Discrepancy relative to $\gamma_2$-norm.} Given an $m\times n$ matrix $A$ and a set $J\subseteq [n]$, let $A_{|J}$ denote the $m\times |J|$ matrix restricted to columns of $A$ in $J$. The
\textit{hereditary discrepancy of $A$} is defined as
\[ herdisc(A) = \max_{J\subseteq[n]}disc(A_{|J}) .\]
Hereditary discrepancy is often a better measure of the complexity of a set system than the discrepancy. It is also rather well behaved; while no polynomial time algorithm can distinguish between set systems with zero discrepancy and set systems with $O(\sqrt{n})$ discrepancy (assuming $P\neq NP$) \cite{CNN11}, hereditary discrepancy can be efficiently approximated. It was shown in \cite{MNT14} that for any matrix $A$, there exists an efficiently computable quantity $\gamma_2(A)$ such that
\begin{equation}
\label{eqn:herdisc}
 \Omega(\gamma_2(A)/\log m)\le herdisc(A) \le O(\gamma_2(A)\sqrt{\log m}) .
 \end{equation}

The proof for the upper bound above was non-constructive in the sense that, given any $J\subseteq [n]$, it was not known how to efficiently find a coloring $x:J\rightarrow\{-1,1\}$ with discrepancy of $A_{|J}$ bounded by the right hand side of \eqref{eqn:herdisc}.
Using Theorem~\ref{thm:main}, we get an algorithm to find such a coloring.
\begin{corollary}
\label{thm:herdisc}
There exists an efficient randomized algorithm that, given any $m\times n$ matrix $A$ and $J\subseteq[n]$, returns a coloring $x:J\rightarrow\{-1,1\}$ such that with constant probability,
\[ \| A_{|J} x\|_{\infty} = O(\gamma_2(A)\sqrt{\log m})  .\]
\end{corollary}
\paragraph{A Generalization of Banaszczyk's result.} We also give a generalization of Theorem~\ref{thm:bana}. Let
$\mathbb{B}_2^m$ denote the Euclidean ball in $\R^m$ of radius $1$ and centred at the origin.
\begin{theorem}
\label{thm:genbana}
Let $S_1,S_2,\dots,S_n$  be sets such that for each $i\in [n]$, $S_i \subseteq \mathbb{B}_2^m$ and $0$ lies in the convex hull of $S_i$. Then for any convex body $K$ with $\gamma_m(K)\ge 1/2$, there exist vectors $v_i\in S_i$ such that
$\sum_{i=1}^n v_i \in cK$,
where $c>0$ is an absolute constant. Moreover, there is an efficient algorithm to find these vectors.
\end{theorem}
Theorem~\ref{thm:bana} is the special case of the above theorem, obtained by taking $S_i=\{-v_i,v_i\}$ for each $i\in [n]$. We will constructively reduce Theorem~\ref{thm:genbana} to Theorem~\ref{thm:bana}, which implies that an algorithm for the latter also gives an algorithm for the former. Similar generalizations for several other problems in discrepancy are mentioned in \cite{B08}.

\subsection{Organization of the Paper}
We state the algorithm for Theorem~\ref{thm:main} in Section~\ref{sec:algorithm}. The analysis is in Section~\ref{sec:analysis}; with a  sketch of the main ideas in Section~\ref{sec:mainideas}. The applications are discussed in Section~\ref{sec:applications}.

\subsection{Acknowledgments}

The authors would like to thank Aleksandar Nikolov 
for innumerable discussions related to the topic of this paper. N.B. and S.L. would like to thank
the Dagstuhl seminar ``Computational Complexity of Discrete Problems" where this work was initiated.

\section{Algorithm Description: The Gram-Schmidt Walk}
\label{sec:algorithm}

The algorithm will proceed in time steps $t=1,2,\ldots,n$, and will maintain a fractional coloring $x_t\in[-1,1]^n$ at all time steps. Let $x_0$ be an arbitrary initial fractional coloring.
Let $x_t$ denote the coloring at the end of time step $t$. An element $i$ is called \textit{alive} at time $t$ if $|x_{t-1}(i)|<1$ and \textit{frozen} (or fixed) otherwise. Let $A_t \subseteq [n]$  denote the set of alive elements at time $t$.


We now give some notation to describe the update step at each time $t$.  Let
$n(t) \in A_t$ denote largest indexed element that is alive at time $t$. We call
$n(t)$ the \textit{pivot} element at time $t$.  Let $V_{t}$ denote the subspace
spanned by the vectors $\{v_i: i\in A_t, i\neq n(t)\}$ and let
$v^{\perp}(t) := \Pi_{V_t^\perp} v_{n(t)}$ denote the projection of $v_{n(t)}$
orthogonal to $V_{t}$.  Note that $v^{\perp}(t)$ depends on $t$ and not just on
$v_{n(t)}$; this is because in our algorithm, $A_t$ (and hence $V_t$) may change
between time steps even if $n(t)$ remains the same. From the perspective of the
walk, $v^{\perp}(t)$ will correspond to the direction that the output of the
random walk will move in during timestep $t$.  To justify the name of the walk,
note that $v^\perp(t)$ is simply the last vector of the Gram-Schmidt
orthogonalization of the ordered sequence of vectors $(v_i)_{i \in A_t}$.



We now describe the algorithm formally. 

\begin{tcolorbox}
\textbf{Algorithm description:}\\
Input: vectors $v_1,\ldots,v_n \in \R^m$ of $\ell_2$ norm at most one; an initial coloring $x_0 \in [-1,1]^n$.\\
Output: coloring $x \in \{-1,1\}^n$.
\begin{enumerate}
\item Given the initial coloring $x_0$, initialize $A_1 = \set{i \in [n]: |x_0(i)| < 1}$, $n(1) = \max \set{i \in A_1}$ and $t=1$.
\item While $A_t \neq \emptyset$, do the following:
\begin{enumerate}
\item Compute an update direction $u_t=(u_t(1),\ldots,u_t(n)) \in \mathbb{R}^n$ as follows:
\begin{itemize}
\item if $i \notin A_t$ set $u_t(i)=0$.
\item if $i=n(t)$ set $u_t(i)=1$.
\item for $i \in A_t \setminus \{n(t)\}$ set $u_t(i)$ to satisfy $v^\perp(t)=v_{n(t)}+\sum_{i \in A_t \setminus \{n(t)\}}  u_t(i) v_i$.
\end{itemize}
\item Let $\delta^{-}_t <0< \delta^{+}_t$ be the unique negative and positive
solutions for $\delta$, respectively, for the equation 
$\max_{i \in A_t} |x_{t-1}(i) + \delta u_t(i)| = 1$.
Update the coloring $x_{t-1}$ randomly as
\[ x_t = x_{t-1} + \delta_t u_t \]
where $\delta_t \in \{\delta^{-}_t,\delta^{+}_t\}$ is chosen randomly as
\[
\delta_t =
\begin{cases}
\delta^{-}_t & \text{with probability } \frac{ \delta^{+}_t}{ \delta^{+}_t -  \delta^{-}_t}\\
\delta^{+}_t & \text{with probability } \frac{ -\delta^{-}_t}{ \delta^{+}_t -  \delta^{-}_t}
\end{cases}\;.
\]
\item Update $A_{t+1} = \{i \in [n]: |x_t(i)| < 1\}$, $n(t+1) = \max \{i \in
A_{t+1}\}$ and $t \leftarrow t+1$.
\end{enumerate}
\item Output $x_t$.
\end{enumerate}
\end{tcolorbox}


\section{Algorithm Analysis}
\label{sec:analysis}

We proceed now to the analysis of our algorithm. In the first subsection we develop some preliminaries which will be helpful later.

\subsection{Preliminaries}
\label{sec:prel}


To bound the discrepancy, we will use a concentration inequality which is a variant of Freedman's inequality for martingales~\cite{Freedman}. The following lemma will be useful.

\begin{lemma}
\label{lem:prel1}
Let $X$ be a random variable such that $X\le 1$. Then for any $\lambda>0$,
\[ \E[e^{\lambda X}] \le \exp\left(\lambda \E[X] +  (e^{\lambda} -\lambda-1 )\E[X^2]\right).\]
\end{lemma}

\begin{proof}
Let $f(x) = \frac{e^{\lambda x} - \lambda x - 1}{x^2}$ where we set $f(0)=\lambda^2/2$. It can be verified that $f(x)$ is increasing for all $x$.
This implies $e^{\lambda x} \leq f(1)x^2 + 1 + \lambda x$ for any $x \le 1$.
Taking expectation, this becomes
\[ \E[e^{\lambda X}] \leq  1+\E[\lambda X]+f(1)\E [X^2]
 =  1+\lambda\E[ X]+(e^{\lambda} -\lambda-1 )\E [X^2]
 \le   e^{\lambda \E [X] +(e^{ \lambda} -\lambda-1 )\E [X^2]} \]
where the last inequality uses the fact that $1+x\le e^x$.
\end{proof}



We will use the following concentration inequality to bound the discrepancy. This is a slight modification of Freedman's inequality given by Yin-Tat Lee and we show its proof below.

\begin{lemma}
\label{lem:prel2}
Let $X_1,\ldots,X_n$ and $Z_0,Z_1,\ldots,Z_n$ be random variables that satisfy
\begin{enumerate}
\item $Z_0$ is deterministic,
\item $Z_{t}-Z_{t-1}\le X_t$ for all $t=1,\ldots,n$ with probability one,
\item $X_t\leq 1$ for all $t=1,\ldots,n$ with probability one, and
\item $\E[X_t+X_t^2|Z_1,\ldots,Z_{t-1}] \le 0$ for all $t=1,\ldots,n$ with probability one.
\end{enumerate}
Then
\[
\E[e^{Z_n}] \leq e^{ Z_0}.
\]
\end{lemma}
\begin{proof}
Let $\lambda>0$ be a real number to be determined later. We shorthand as $\E_{t-1}[\cdot]$ the conditional expectation $\E[\cdot | Z_1,\dots, Z_{t-1}]$.
We first bound $\E_{t-1}[e^{\lambda Z_t}]$ for which we observe the following:
\begin{eqnarray*}
\E_{t-1}\left[e^{\lambda Z_t}\right] &= & e^{\lambda Z_{t-1}}\E_{t-1}\left[e^{\lambda(Z_t-Z_{t-1})}\right] \\
&\le & e^{\lambda Z_{t-1}}\E_{t-1}\left[e^{\lambda X_t}\right] \\
& \le &  e^{\lambda Z_{t-1}} \exp\left(\lambda \E_{t-1}[X_t] +(e^{ \lambda} -\lambda-1 )\E_{t-1}\left[X_t^2\right]\right) \qquad \textrm{(using Lemma~\ref{lem:prel1})} \\
& \le & e^{\lambda Z_{t-1}}\exp\left((e^{ \lambda} -2\lambda-1 )\E_{t-1}\left[X_t^2\right]\right) \qquad \textrm{(using }\E_{t-1}[X_t] \le -\E_{t-1}[X_{t}^2] )\\
\end{eqnarray*}
Set $\lambda>0$ to be the (unique) solution of $e^{ \lambda} -2\lambda-1 =0.$
As $e^x \leq 1 + x + x^2$ for $x \leq 1$, for $0<\lambda<1$ we have
\[ e^{\lambda} -2\lambda-1 \le\lambda^2 -\lambda <0. \]
Thus, it must hold that $ \lambda \geq 1 $. Then we get $\E_{t-1}[e^{\lambda Z_t}] \le e^{\lambda Z_{t-1}} $. And thus by
induction,  $\E[e^{\lambda Z_n}] \le \E[e^{\lambda Z_0}] = e^{\lambda Z_0}$, since
$Z_0$ is deteministic. Since $\lambda \geq 1$, by Jensen's inequality
$\E[e^{ Z_n}] \leq \E[e^{\lambda Z_n}]^{1/\lambda} \leq e^{ Z_0}$,
as needed.
\end{proof}

\subsection{Notations and Preliminary Observations}
We can now start with the analysis of the algorithm. Since at the end of each
time step $t$, at least one more element gets colored $-1$ or $1$, the algorithm
clearly terminates in at most $n$ steps. To simplify notations, if the algorithm
terminates after $t<n$ steps and outputs $x_t$, we set $x_{t+1}=\ldots=x_n := x_t$.
As we maintain $\| x_t\|_{\infty} \leq 1$ for all time steps $t$, we see that $x_n\in \{-1,1\}^n$.


It should also be noted in step $(2.a)$ of the algorithm that such a direction $u_t$ always exists. This is because $v_{n(t)}-v^{\perp}(t)$ lies in the subspace $V_t$ and hence there always exist $u_t(i)$'s such that
\[ v_{n(t)}-v^{\perp}(t) = -\sum_{i\in A_t \setminus \{n(t)\}} u_t(i) v_i .\]
This then gives $v^{\perp}(t) = v_{n(t)} + \sum_{i\in A_t \setminus \{n(t)\}} u_t(i) v_i = \sum_{i \in A_t}  u_t(i) v_i$. As $u_t(i)=0$ if $i \notin A_t$, we obtain that
\begin{equation}\label{eq:disc_update_t}
v^{\perp}(t)=\sum_{i=1}^n u_t(i) v_i.
\end{equation}

We now focus on showing the subgaussianity bound in Theorem~\ref{thm:main}.
Henceforth, we fix a vector $\theta \in \R^m$ with respect to which we want to show
subgaussianity. Let $A$ be the $m \times n$ matrix whose columns are given by
$v_1,\ldots,v_n$. Define $Y := \sum_{i=1}^n (x(i)-x_0(i))v_i$ and let
\[disc(\theta)=\langle \theta,Y \rangle =   \langle \theta,\sum_{i=1}^n (x(i)-x_0(i)) v_i \rangle = \sum_{i=1}^n (x(i)-x_0(i)) \langle \theta,v_i \rangle \]
 and hence $\langle \theta,Y \rangle $ can be seen as the discrepancy of the ``row" $\theta^TA$, which has $\langle \theta,v_i \rangle$ as its $i^{th}$ entry.

Let us denote the respective signed discrepancy
at the end of time step $t$ by
\[
disc_t := \langle \theta, \sum_{i=1}^n  (x_t(i)-x_0(i)) v_i\rangle =
\sum_{i=1}^n  (x_t(i)-x_0(i))\langle \theta,v_i  \rangle,
\]
let
\[
\Delta_t x := x_t - x_{t-1} = \delta_t u_t
\]
denote the coloring update at time $t$, and let
\[
\Delta_t disc := disc_t - disc_{t-1} = \sum_{i=1}^n \langle \theta,v_i  \rangle \Delta_t x(i)
= \delta_t \sum_{i=1}^n \langle \theta, v_i \rangle u_t(i)
\]
be the change in discrepancy at time $t$. Our end goal
is to show that $\E[e^{\disc_n}] \leq e^{(\sigma^2/2) \|\theta\|_2^2}$, for $\sigma
= O(1)$ to be computed later. 

A key observation is that the change in discrepancy at time $t$ depends only on
the vector $v^{\perp}(t)$.


\begin{lemma}
\label{lem:changedisc} At each time step $t$, $\Delta_t disc = \delta_t \langle \theta, v^{\perp}(t)\rangle$.
\end{lemma}
\begin{proof}
The change in discrepancy at time $t$ is
\[\Delta_t \disc = \delta_t \sum_{i=1}^n \langle \theta,v_i \rangle u_t(i) =
\delta_t \langle \theta,  \sum_{i=1}^n u_t(i)v_i \rangle = \delta_t \langle
\theta, v^\perp(t) \rangle,
\]
where the last equality follows by~\eqref{eq:disc_update_t}.
\end{proof}



\subsection{Main ideas}
\label{sec:mainideas}

Before we give the technical details, we describe the main ideas of the analysis, which are simple, and then give a roadmap of the analysis.

First, as our algorithm can start with any initial coloring $x_0$, we may assume without loss of generality that the vectors $v_1,\ldots,v_n$ are linearly independent.
Otherwise, we can apply a standard preprocessing step that removes linear dependencies, by finding some linear dependency among the vectors and
making an update step which incurs zero discrepancy. 
 Our algorithm can be viewed as a randomized extension of the above dependent rounding approach.

Suppose for a moment that the element that got colored at each time step was the pivot. That is, the elements got colored in the order $n,n-1,\ldots,1$. Then, at time $t$, the pivot is $n(t) = n-t+1$ and the vectors $v^{\perp}(t)$
can be described as follows. Let $w_1,\ldots,w_n$ be the orthonormal vectors
obtained by applying the Gram-Schmidt orthonormalization procedure (GS) on the vectors $v_1,\ldots,v_n$ in that order. That is, $w_1=v_1/\|v_1\|$ and for $i>1$, $w_i$ is the projection
of $v_i$ orthogonal to $v_1,\ldots,v_{i-1}$, normalized to have unit norm.
By our assumption that $v_i$ are linearly
independent, each $w_i$ is non-zero. It is easily checked that $v^{\perp}(t) =
 \langle v_{n(t)},w_{n(t)} \rangle w_{n(t)}$.  Another observation that will be
useful later is that $w_{n(t)}$ (and hence $v^{\perp}(t)$) depends only on the
set $\{v_1,\ldots,v_{n(t)-1}\}$ and not the particular order in which GS is
applied to this set.

Now $\delta_t$ is a mean-zero random variable which is independently chosen at each time $t$. Also $|\delta_t|\leq 2$ for all $t$ (see Lemma~\ref{lem:sumofdelta}). This suggests that the moment generating function of the discrepancy is
\[
\E\left[e^{disc(\theta)}\right]  = \E\left[e^{ \sum_{t=1}^n  \delta_t \langle  \theta, v^{\perp}(t)\rangle  }\right]
\le  e^{O(1) \cdot \sum_{t=1}^n \langle \theta, v^{\perp}(t) \rangle^2} . \]
But this is at most $ e^{O(1) \cdot \| \theta\|_2^2}$ as
\[ \sum_t \langle \theta, v^{\perp}(t) \rangle^2  = \sum_t \langle  \theta,  \langle v_{n(t)},w_{n(t)} \rangle w_{n(t)} \rangle^2 \leq  \sum_t \langle  \theta, w_{n(t)} \rangle^2 \leq \| \theta\|_2^2 .\]
Here the first inequality follows as $|\langle v_{n(t)},w_{n(t)} \rangle|\leq 1$ as the $w_i$'s are orthonormal unit vectors and $\|v_i\|_2 \leq 1$ for each $i$.
The second inequality follows as $\sum_i \langle  \theta, w_{i} \rangle^2 \leq \|\theta\|^2_2$ for any orthonormal collection of vectors $w_i$.

There are two issues that need to be addressed in the simplified description above.
First, non-pivot elements may get colored sometimes. That is, the variables may not get colored in the order $n,n-1,\ldots,1$. The key point of the analysis is to show that this only makes the problem easier. To make this a bit more precise, let us view
$\sum_i \langle  \theta, w_{i} \rangle^2 \leq \|\theta\|^2_2$  as the energy budget initially available to us. If some non-pivot element $x_k$ is colored at some time $t$, then the GS procedure (without $v_k$) will produce a different set of orthonormal vectors $\{w'_i\}$. However we can bound the increase
\[\langle \theta,w'_{n(t)} \rangle^2 - \langle \theta,w_{n(t)} \rangle^2 \]
 in the pivot's energy by the amount  $\langle  \theta, w_{k} \rangle^2$ which
was available to us, but we will never use it anymore as $k$ will never be a pivot once it is colored.
The formal analysis later on will divide the time steps into \textit{phases} where the pivot element remains the same during a phase and consider the evolution of $v^{\perp}(t)$ in each phase.

A second technical issue is that to obtain a subgaussian distribution for all $\theta \in \mathbb{R}^m$, we need to control the variance of the energy which requires that $|\langle \theta, v^{\perp}(t \rangle)|=O(1)$. Let us refer to the times when this does not happen as bad times. To get around this issue we break the discrepancy contribution into two parts: a deterministic part due to the bad times, and a random contribution due to the other time steps. Again by considering the dynamics of how $v^{\perp}(t)$ evolves during a phase, one can show that the cumulative discrepancy due to all the bad times can be at most $O(\| \theta\|_2^2)$.


\subsection{Phases and Dynamics of Pivots}
To analyze the process, we need to set up some more notation. Let $T$ be the first time at the end of which all the elements are colored. Recall that at the
beginning of each time step $t \in [T]$, there is some pivot element $n(t)$ and
some set of alive elements $A_t$. Recall that $V_t = {\rm span}\set{v_i:i\in A_t, i\neq n(t)}$. After the update at time $t$, some element is frozen. By
convention, we define $V_0 := {\rm span}\set{v_1,\ldots,v_n}$. Note that for any time $t \in
[T]$, we have $V_{t} \subseteq V_{t-1}$.


We divide the time steps $t \in [T]$ into $\kappa \in [T]$ disjoint phases,
where a phase is a maximal sequence of time steps with the same pivot element
(note that $\kappa$ is a random variable). Let $t^b_k$ be the time when phase $k
\in [\kappa]$ begins and $t^e_k$ be the time when this phase ends. Note that
$n(t)=n(t^b_k)$ for all $t \in [t^b_k,t^e_k]$ and the pivot
element $n(t^b_k)$ gets frozen after the update at time step $t^e_k$. Thus, the first
phase, with $n(1)$ as pivot, begins at the beginning of time step $t^b_1=1$ and ends at the end of time step $t^e_1=t^b_2-1$.
Given a time $t \in [T]$, we define $f_t = \min \set{t' \in [t]: n(t') = n(t)}$ to
be the beginning time of the phase that $t$ belongs to. Note that if $t \in
[t^b_k,t^e_k]$ then $f_t = t^b_k$.

We will now give a useful characterization of the discrepancy vector
$v^\perp(t)$ at any time step $t \in [T]$. For any subspace $W \subseteq \R^m$, let $W^{\perp}$ denote the subspace orthogonal to $W$ and
 let $\Pi_{W}(\cdot)$ denote the projection operator on $W$. For two linear subspaces $W_1
\supseteq W_2$, we use the notation $W_1 / W_2 := W_1 \cap W_2^\perp$.

\begin{lemma}
\label{lem:perpt}
At each time step $t \in [T]$ of the algorithm the following holds:
\begin{equation}
\label{eq:perpt}
v^{\perp}(t) = \sum_{i=f_t}^{t} \Pi_{V_{i-1}/V_i}(v_{n(f_t)}),
\end{equation}
where the subspaces $V_{i-1}/V_i$ are mutually orthogonal.
\end{lemma}

Lemma~\ref{lem:perpt} follows directly from the following useful fact.

\begin{lemma}
\label{lem:decomp}
Let $\R^m \supseteq W_0 \supseteq W_1 \supseteq \cdots \supseteq W_t $, $t \geq
1$, denote a non-increasing sequence of linear subspaces and let $y \in W_0$.
Then
\begin{equation}
\label{eq:decomp}
\Pi_{W_t^\perp}(y) = \sum_{i=1}^t \Pi_{W_{i-1} / W_i}(y),
\end{equation}
where the subspaces $W_{i-1}/W_i$ for $i\in [t]$ are mutually orthogonal and are also orthogonal to $W_t$.
\end{lemma}
\begin{proof}
We first prove orthogonality. Let $W_{t+1}=\emptyset$ denote the empty subspace. Then $W_t = W_t / W_{t+1}$. Now for $1 \leq i < j \leq t+1$, we need to show that
$W_{i-1}/W_i$ and $W_{j-1}/W_j$ are orthogonal. This follows
directly since $W_{i-1}/W_i$ is orthogonal to the subspace $W_i$ and
$W_{j-1}/W_j \subseteq W_{j-1} \subseteq W_i$ since $j-1 \geq i$.

To prove \eqref{eq:decomp}, notice that ${\rm span}\{W_i\cup (W_{i-1}/ W_i)\} =W_{i-1}$ and thus by induction,
\[W_0={\rm span} \{ \cup_{i=1}^{t+1} (W_{i-1}/ W_i )\} \]
where each of the subspaces $W_{i-1}/ W_i$ are mutually orthogonal. This implies that for $y\in W_0$,
\[ y =\Pi_{W_0}(y)= \sum_{i=1}^{t+1} \Pi_{W_{i-1} / W_i}(y)  = \sum_{i=1}^{t} \Pi_{W_{i-1} / W_i}(y) +  \Pi_{W_{t}}(y) .\]
The lemma follows now by observing that $\Pi_{W_t^{\perp}} ( \Pi_{W_{t}}) =0 $ and $\Pi_{W_t^{\perp}}(  \Pi_{W_{i-1} / W_i}  )  = \Pi_{W_{i-1} / W_i}$ as $W_{i-1} / W_i \subseteq W_t^{\perp}$.
%
\end{proof}

\begin{proof}[Proof of Lemma~\ref{lem:perpt}]
First, by definition of $f_t$, note that $n(t) = n(f_t)$ and hence
\[
v^\perp(t) = \Pi_{V_t^\perp}(v_{n(t)}) = \Pi_{V_t^\perp}(v_{n(f_t)}) .
\]

We now check that $v_{n(f_t)} \in V_{f_t-1}$. If $f_t = 1$, this is immediate
since $V_0 = {\rm span}\set{v_1,\ldots,v_n}$ by convention. Otherwise, since $n(f_t)$ is both alive and
not the pivot at time $f_t-1 \geq 1$, we have that $n(f_t) \in A_{f_t-1}
\Rightarrow v_{n(f_t)} \in V_{f_t-1}$. The main statement now follows directly by
applying Lemma~\ref{lem:decomp} on the subspaces $V_{f_t-1} \supseteq V_{f_t}
\supseteq \cdots \supseteq V_t$ and the vector $v_{n(f_t)}$.
\end{proof}



\subsection{Discrepancy in a Phase}
We now bound the discrepancy incurred during any subinterval of a phase. We first need the following simple but very useful fact.
\begin{lemma}
\label{lem:sumofdelta}
Let $p \le q$ be any two time steps in $[t^b_k, t^e_k]$ during a phase $k \in [\kappa]$.
Then $ \left|\sum_{t=p}^{q} \delta_{t}\right| \le 2.$
\end{lemma}
\begin{proof}
At any time $t \in [p,q]$, note that $n(t) = n(p)$, i.e.~the pivot remains
unchanged. The color of the pivot element $n(p)$ at time $t \in [p,q]$ is
updated by $\delta_t u_t(n(p)) = \delta_t$ and hence
\[ x_{q}(n(p)) - x_{p-1}(n(p)) = \sum_{t=p}^{q} \delta_{t} . \]
As $|x_t(n(p))| \le 1$ for all $t \in [T]$, we have that $|x_{q}(n(p)) -
x_{p-1}(n(p))| \le 2$ as needed.
\end{proof}

\begin{lemma}
\label{lem:discphase}
Let $p \le q$ be any two time steps in $[t^b_k, t^e_k]$ during a phase $k \in [\kappa]$.
The discrepancy
$|\disc_{q}-\disc_{p-1}|$ incurred during the time interval $[p,q]$ is
at most $2 \|\theta^{(k)}\|_2$, where $\theta^{(k)} := \Pi_{V_{t^b_k-1} / V_{t^e_k}}(\theta)$.
Furthermore, the subspaces $V_{t^b_k-1} / V_{t^e_k}$ for $k \in [\kappa]$ are mutually orthogonal.
\end{lemma}
\begin{proof}
We first prove orthogonality. Take $k_1,k_2 \in [\kappa]$, where $k_1 < k_2$. We
must show that the subspaces $V_{t^b_{k_1}-1} / V_{t^e_{k_1}}$ and $V_{t^b_{k_2}-1} /
V_{t^e_{k_2}}$ are orthogonal. This follows since the first subspace is
orthogonal to $V_{t^e_{k_1}}$ and the second subspace is contained in
$V_{t^b_{k_2}-1} \subseteq V_{t^e_{k_1}}$ since $t^b_{k_2}-1 \geq t^e_{k_1}$. We now prove the main statement. First, note that
\begin{equation}
\label{eq:dp1}
disc_{q}-disc_{p-1} = \sum_{t=p}^{q} \delta_t \ip{\theta}{v^\perp(t)}\text {.}
\end{equation}
By Lemma~\ref{lem:perpt}, for $t \in [p,q] \subseteq
[t^b_k,t^e_k]$, letting $h := n(t^b_k) = n(t)$ denote the pivot index, we have that
\begin{equation}
\label{eq:dp2}
v^\perp(t) := \Pi_{V_t^\perp}(v_h) = \sum_{i=t^b_k}^{t} \Pi_{V_{i-1}/V_i}(v_h) .
\end{equation}
Combining~\eqref{eq:dp1},~\eqref{eq:dp2} we get that
\begin{eqnarray}
\left|disc_{q}-disc_{p-1}\right| &=& \left|\sum_{t=p}^{q} \delta_t \cdot
\sum_{i=t^b_k}^t \ip{\Pi_{V_{i-1}/V_i}(v_h)}{\theta}\right| \nonumber\\
   &=& \left|\sum_{i=t^b_k}^{q} \ip{\Pi_{V_{i-1}/V_i}(v_h)}{\theta} \cdot \sum_{t= \max \set{i,p}}^{q} \delta_t
\right| \nonumber\\
   &\leq& \sum_{i=t^b_k}^{q} \left|\ip{\Pi_{V_{i-1}/V_i}(v_h)}{\theta}\right| \cdot \left|\sum_{t= \max \set{i,p}}^{q} \delta_t\right| \nonumber
 \\
   &\leq & 2 \sum_{i=t^b_k}^{q} \left|\ip{\Pi_{V_{i-1}/V_i}(v_h)}{\theta}\right| \quad
\left(\text{ by Lemma~\ref{lem:sumofdelta} }\right)\nonumber \\
   &=& 2 \sum_{i=t^b_k}^{q} \left|\ip{\Pi_{V_{i-1}/V_i}(v_h)}{\Pi_{V_{i-1}/V_i}(\theta)}\right| . \label{eq:dp3}
\end{eqnarray}
Now applying Cauchy-Schwarz inequality twice, we get
\begin{eqnarray}
\sum_{i=t^b_k}^{q} \left|\ip{\Pi_{V_{i-1}/V_i}(v_h)}{\Pi_{V_{i-1}/V_i}(\theta)}\right|
&\leq &\sum_{i=t^b_k}^{q}
\|\Pi_{V_{i-1}/V_i}(v_h)\|_2 \|\Pi_{V_{i-1}/V_i}(\theta)\|_2 \nonumber \\
&\leq &\left(\sum_{i=t^b_k}^{q}
\|\Pi_{V_{i-1}/V_i}(v_h)\|_2^2\right)^{1/2} \cdot \left(\sum_{i=t^b_k}^{q} \|\Pi_{V_{i-1}/V_i}(\theta)\|_2^2\right)^{1/2}\nonumber \\
&\leq &\|\Pi_{V_{t^b_k-1}/V_{q}}(v_h)\|_2 \cdot \|\Pi_{V_{t^b_k-1}/V_{q}}(\theta)\|_2\nonumber \\
&\leq & \|\Pi_{V_{t^b_k-1}/V_{q}}(\theta)\|_2 \quad \left(~\text{since } \|v_h\|_2 \leq 1 ~\right) \nonumber \\
&\leq & \|\Pi_{V_{t^b_k-1}/V_{t^e_k}}(\theta)\|_2 \quad \left(~\text{since } V_q\supseteq V_{t^e_k} ~\right) . \label{eq:dp4}
\end{eqnarray}
where the third inequality follows by orthogonality of the subspaces
$V_{i-1}/V_i$ for all $i \in [t^b_k,t^e_k]$, and their containment inside
$V_{t^b_k-1}/V_{q}$. The lemma now follows by combining~\eqref{eq:dp3},~\eqref{eq:dp4}.
\end{proof}

Notice that discrepancy $\{\disc_t: t \ge 0\}$ is a martingale. For the rest of the
analysis, we will define another closely related martingale $\{Y_t\}$ and show
that the subgaussianity of $\{\disc_t\}$ follows from the subgaussianity of $\{Y_t\}$. 
Then, we will show the subgaussianity of $\{Y_t\}$.
Define the random process $\{Y_t:  t\ge 0\}$ as
\begin{align*}
Y_t :=
	\begin{cases}
		0  & \mbox{if } t=0 \\
		Y_{t-1} & \mbox{if } \| \Pi_{V_{f_t-1}/V_{t}} (\theta) \|_2 > 1/8 \\
		Y_{t-1}+ \Delta_t \disc & \mbox{otherwise}
	\end{cases}
\; .
\end{align*}
Note that $\E[Y_t|Y_1,\ldots,Y_{t-1}] = Y_{t-1}$ and thus $\{Y_t\}$ is a martingale.
To prove Theorem~\ref{thm:main}, we will bound the exponential moment of $disc_n$.

\begin{theorem}
\label{thm:moment-bound}
$ \E[e^{\disc_n}] \le e^{20 \|\theta\|_2^2}. $
\end{theorem}
Note that this directly gives that the walk is $\sqrt{40} \approx 6.32$ subgaussian.

To control this exponential moment, we will split $\disc_n=\disc_T = (\disc_T-Y_T) +
Y_T$, where $(\disc_T-Y_T)$ will correspond to the discrepancy incurred during
``bad'' times and $Y_T$ corresponds to the ``good'' times. More precisely,  call a time step $t$
\textit{good} if $\| \Pi_{V_{f_t-1}/V_{t}} (\theta) \|_2 \le 1/8$ and
\textit{bad} otherwise. Let
\[
\mathcal{B} := \set{t \in [T]: \|\Pi_{V_{f_t-1}/V_t}(\theta)\|_2 > 1/8}
\]
denote the set of bad times. Note that
\[
\disc_T = Y_T+ \sum_{t \in \mathcal{B}} \Delta_t \disc.
\]
The following lemma now shows that the discrepancy incurred during the bad times
can be upper bounded deterministically.

\begin{lemma} $
\left|\sum_{t \in \mathcal{B}} \Delta_t \disc\right| < 16 \|\theta\|_2^2.
$
\label{lem:l1}
\end{lemma}
\begin{proof}

For each $k \in [\kappa]$, let $\mathcal{B}_k$ denote the set of bad time steps in phase $k$. Notice that in a phase, once a time step becomes bad, all subsequent time steps in that phase are bad as the length of
$\Pi_{V_{f_t-1}/V_{t}} (\theta)$ is a non-decreasing function of $t$ during a
phase (as $V_{t+1} \subseteq V_{t}$ for every $t$).  Thus, the set $\mathcal{B}_k$, if non-empty, forms a consecutive interval (ending at the end of phase $k$), and hence
Lemma~\ref{lem:discphase} can be applied to it.

Recall that we denoted $\theta^{(k)}:=\Pi_{V_{t^b_k-1}/V_{t^e_k}} (\theta)$ for $k \in [\kappa]$. Note that the vectors
$\theta^{(1)},\dots,\theta^{(\kappa)}$ are mutually orthogonal. Furthermore, by the argument above, for any $k \in [\kappa]$ with $\mathcal{B}_k\neq \emptyset$ we have that
$\|\theta^{(k)}\| \geq 1/8$.



The discrepancy incurred during bad times can now be bounded as follows
\begin{align*}
\left|\sum_{t \in \mathcal{B}} \Delta_t \disc\right|
&\leq \sum_{k \in [\kappa]} \left|\sum_{t \in \mathcal{B}_k} \Delta_t \disc\right| \\
&\leq \sum_{k \in [\kappa], \mathcal{B}_k\neq \emptyset} 2 \|\theta^{(k)}\|_2
\quad \left( \text{ by Lemma~\ref{lem:discphase} } \right) \\
& \leq  \sum_{k \in  [\kappa], \mathcal{B}_k\neq \emptyset} 16 \|\theta^{(k)}\|^2_2 \quad \left(\|\theta^{(k)}\| \geq 1/8 \text{ if } \mathcal{B}_k\neq \emptyset  \right) \\
&\leq 16 \|\theta\|^2_2 \quad \left( \text{ by orthogonality
of } \theta^{(1)},\dots,\theta^{(\kappa)} \right). \qedhere 
\end{align*}
\end{proof}

\subsection{Bounding Discrepancy}
We now prove Theorem \ref{thm:moment-bound}. To this end, we will define potential functions that capture how the variance of $Y_t$ increases over time.
%
%
%
%
 For $t \in [T]$, define the subspace $R_t := {\rm span}\{v_i: i \in A_t\}$,
i.e.~the span of the active vectors. Note that at any timestep $t \in [T]$, the
discrepancy direction $v^\perp(t) \in R_t / V_t$. Consider the potentials

\begin{align*}
\Phi_t^b :=
	\begin{cases}
		 \|\Pi_{V_t}(\theta)\|^2 + (1-x_{t-1}(n(t))^2) \|\Pi_{R_t/V_t}(\theta)\|_2^2  & \mbox{if }t \text{ is good}  \\
		\|\Pi_{V_t}(\theta) \|^2  &\mbox{if } t \text{ is bad}
	\end{cases}
\end{align*}
and
\begin{align*}
\Phi_t^e :=
	\begin{cases}
		 \|\Pi_{V_t}(\theta)\|^2 + (1-x_t(n(t))^2) \|\Pi_{R_t/V_t}(\theta)\|_2^2  & \mbox{if }t \text{ is good}  \\
		 \|\Pi_{V_t}(\theta) \|^2  &\mbox{if } t \text{ is bad}
	\end{cases}
\; .
\end{align*}
Here we think of $\Phi_t^b$ as the potential at the beginning of time step $t$ and
$\Phi_t^e$ as the potential at the end of time step $t$. By convention, we define
$\Phi_0^e = \Phi_0^b = \|\theta\|_2^2$.

As by Lemma~\ref{lem:l1} we have deterministic control over the discrepancy
incurred during the bad times, our goal is now to control the discrepancy
during the good times. Namely, we need to bound the process
$Y_0,Y_1,\dots,Y_T$. Here, the main idea is to charge the drop in potential to the increase in discrepancy. To this end, we define the potential weighted
discrepancy process $\{Z_t\}$ for $t\ge 0$ as
\[
Z_t := Y_t + 4 \Phi^e_t .
\]
Note that $Y_0=0$, $Z_0 = 4\Phi^e_0 = 4\|\theta\|_2^2$ and $Z_t \geq Y_t$ for all $t \geq 0$. 
For $t\in[T+1,n]$ we define $Z_t:=Z_{T}$. Our goal is to now show that the process $Z_0,\dots,Z_T$ has
a strong ``negative'' drift.


The increments of $Z_t-Z_{t-1}$ for $t \in [T]$ can be expressed as follows:
\begin{equation}
\label{eq:increment}
Z_t-Z_{t-1} = Y_t-Y_{t-1} + 4(\Phi^e_t-\Phi^e_{t-1})
            = (Y_t-Y_{t-1} + 4(\Phi^e_t-\Phi^b_t)) + 4(\Phi^b_t-\Phi^e_{t-1}) .
\end{equation}
We decompose this increment into a ``predictable'' part
\begin{align}
\label{eq:predicable}
X_t := Y_t-Y_{t-1} + 4(\Phi^e_t-\Phi^b_t)
&= \begin{cases} \delta_t \ip{v^\perp(t)}{\theta}
- 4 (x_{t}(n(t))^2-x_{t-1}(n(t))^2) \|\Pi_{R_t/V_t}(\theta)\|_2^2 &\mbox{if }t \text{ is good} \\ 0 &\mbox{if }t \text{ is bad}
\end{cases} \nonumber \\
&= \begin{cases} \delta_t \ip{v^\perp(t)}{\theta}
- 4\delta_t(\delta_t + 2x_{t-1}(n(t))) \|\Pi_{R_t/V_t}(\theta)\|_2^2 &\mbox{if }t \text{ is good} \\ 0 &\mbox{if }t \text{ is bad}
\end{cases}
\end{align}
over which we will be able to get stochastic control,
and a ``free" part
\[ 4(\Phi^b_t-\Phi^e_{t-1})\]
which we show is always non-positive.  The following
crucial lemma shows that $4(\Phi^b_t-\Phi^e_{t-1})$ indeed gives us a ``free
drop'' in potential.

\begin{lemma}
\label{lem:freedrop}
For all $t \in [T]$, $\Phi^b_t \leq \Phi^e_{t-1}$. Hence, the increments satisfy $Z_t -Z_{t-1} \leq X_t$ for all $t \in [T]$.
\end{lemma}
\begin{proof}
For $t=1$, the statement is trivial since $\Phi^e_0 = \|\theta\|_2^2$ and
clearly $\Phi^b_t \leq \|\theta\|_2^2$ for all $t$. Thus, we may assume $t \geq 2$.
If $t$ is bad, then using $V_t \subseteq V_{t-1}$ we get $\Phi_{t}^b = \|\Pi_{V_{t}}(\theta)\|_2^2 \leq
\|\Pi_{V_{t-1}}(\theta)\|_2^2 \leq \Phi_{t-1}^e$. So, we may assume from now on that $t$ is good. 

If the time step $t$ is the beginning of a new phase, then
\begin{align*}
\Phi^b_t &= \|\Pi_{V_t}(\theta)\|_2^2 + (1-x_{t-1}(n(t))^2) \|\Pi_{R_{t}/V_{t}}(\theta)\|_2^2 \\
         &\leq \|\Pi_{V_{t}}(\theta)\|_2^2 + \|\Pi_{R_{t}/V_{t}}(\theta)\|_2^2 \\
         &= \|\Pi_{R_t}(\theta)\|_2^2 \leq \|\Pi_{V_{t-1}}(\theta)\|_2^2 \leq \Phi^e_{t-1} .
\end{align*}
Here the inequality $\|\Pi_{R_{t}}(\theta)\|_2 \leq \|\Pi_{V_{t-1}}(\theta)\|_2$
follows by our assumption that $t$ is the beginning of a new phase and hence $A_t \subseteq A_{t-1} \setminus \{n(t-1)\}$ i.e. $R_{t} \subseteq V_{t-1}$.

Lastly, if $t$ is good and not the beginning of a new phase, then note that $t-1$ is also good and that $n(t) = n(t-1)$.  Thus
\begin{align*}
\Phi^b_t &= \|\Pi_{V_{t}}(\theta)\|_2^2 + (1-x_{t-1}(n(t))^2)\|\Pi_{R_{t}/V_{t}}(\theta)\|_2^2 \\
   &\leq \|\Pi_{V_{t}}(\theta)\|_2^2 + (1-x_{t-1}(n(t))^2)\|\Pi_{R_{t-1}/V_t}(\theta)\|_2^2
      \quad \left( \text{ using } R_{t} \subseteq R_{t-1}~\right) \\
    &= \|\Pi_{V_{t}}(\theta)\|_2^2 +
    (1-x_{t-1}(n(t))^2)\left(\|\Pi_{V_{t-1}/V_t}(\theta)\|_2^2 + \|\Pi_{R_{t-1}/V_{t-1}}(\theta)\|_2^2\right) \\
    &\leq \left(\|\Pi_{V_{t}}(\theta)\|_2^2 + \|\Pi_{V_{t-1}/V_{t}}(\theta)\|_2^2\right) +
          (1-x_{t-1}(n(t))^2) \|\Pi_{R_{t-1}/V_{t-1}}(\theta)\|_2^2 \\
    &= \|\Pi_{V_{t-1}}(\theta)\|_2^2 + (1-x_{t-1}(n(t))^2)
\|\Pi_{R_{t-1}/V_{t-1}}(\theta)\|_2^2 = \Phi_{t-1}^e 
\end{align*}
as needed.
\end{proof}

We now show the increment bounds $X_t$ satisfy the negative drift
requirements of Lemma~\ref{lem:prel2}.

\begin{lemma}
\label{lem:inc-reqs}
For $t \in [T]$, $|X_t| \leq 1$ and $\E[X_t+X_t^2|Z_1,\ldots,Z_{t-1}] \leq 0$.
\end{lemma}

\begin{proof}
Clearly, we may assume that $t$ is good, since otherwise $X_t = 0$ and the statement
is trivial. For simplicity of notation, let $\Omega_{t-1}$ denote all the random choices made
by the algorithm in the first $t-1$ time steps. Note that $\Omega_{t-1}$ determines in particular $Z_1,\ldots,Z_{t-1}$.
We shorthand $\E_{t-1}[\cdot] := \E[\cdot|\Omega_{t-1}]$. Let us further denote
$\theta_t := \ip{v^\perp(t)}{\theta}$, $\bar{\theta}_t =
\|\Pi_{R_t/V_t}(\theta)\|_2$, and $x := x_{t-1}(n(t))$. With this notation, we
have that
\[
X_t = \delta_t \theta_t - 4 \delta_t(\delta_t+2x) \bar{\theta}_t^2 .
\]
Since $v^\perp(t) \in R_t/V_t$, note that $|\theta_t| \leq \|v^\perp(t)\|_2
\bar{\theta}_t \leq \bar{\theta}_t$. By definition of $t$ being good, we have that $\bar{\theta}_t
\leq 1/8$. Since $x \in [-1,1]$ and $x+\delta_t \in [-1,1]$, we deduce the
following simple bounds
\begin{equation}\label{eq:simplebounds}
|\delta_t| \leq 2, \quad |\delta_t(\delta_t+2x)| = |(\delta_t+x)^2-x^2| \leq 1,
\quad |\delta_t+2x| = |(\delta_t+x) + x| \leq 2 .
\end{equation}

Using \eqref{eq:simplebounds}, we see that
\[
|X_t| \leq |\delta_t||\theta_t| + 4|\delta_t(\delta_t+2x)|\bar{\theta}_t^2
      \leq 2(1/8) + 4(1)(1/8^2) \leq 1  .
\]
Next, we have that $\E_{t-1}[X_t] = - 4 \E_{t-1}[\delta_t^2] \bar{\theta}_t^2$
since $\E_{t-1}[\delta_t] = 0$. Lastly, using the inequality $(a+b)^2 \leq
2a^2+2b^2$, we get that
\begin{align*}
X_t^2 &\leq 2 \delta_t^2 \theta_t^2 + 2(16)\delta_t^2(\delta_t+2x)^2  \bar{\theta}_t^4 \\
      &\leq 2 \delta_t^2 \bar{\theta}_t^2 + 2(16) \delta_t^2 (4)(1/8)^2 \bar{\theta}_t^2
      = 4 \delta_t^2 \bar{\theta}_t^2 .
\end{align*}
Thus $\E_{t-1}[X_t^2] \le 4 \E_{t-1}[\delta_t^2] \bar{\theta}_t^2$ and hence
$\E_{t-1}[X_t+X_t^2] \le 0$. As $Z_1,\ldots,Z_{t-1}$ are determined by $\Omega_{t-1}$,
this implies that $\E[X_t+X_t^2|Z_1,\ldots,Z_{t-1}] \leq 0$, as needed.
\end{proof}

We now prove the main moment bound.

\begin{proof}[Proof of Theorem~\ref{thm:moment-bound}]
To begin we see that
\begin{align*}
\E\left[e^{\disc_n}\right] &\leq \E\left[e^{Y_T + |\sum_{t \in \mathcal{B}} \Delta_t \disc|}\right] \\
          &\leq \E\left[e^{Y_T + 16\|\theta\|_2^2}\right] \quad \left(\text{ by Lemma~\ref{lem:l1} } \right) \\
          &\leq \E\left[e^{Z_n}\right] e^{16\|\theta\|_2^2} \quad \left(\text{ as } Y_T \leq Z_T=Z_n~\right) .
\end{align*}
Recall that $Z_0 = 4\|\theta\|_2^2$. Furthermore, by Lemma~\ref{lem:freedrop} we have
$Z_t-Z_{t-1} \leq X_t$, by Lemma~\ref{lem:inc-reqs} we have $|X_t| \leq 1$ and
$\E\left[X_t+X_t^2|Z_1,\ldots,Z_{t-1}\right] \leq 0$ for all $t \in [T]$, and for $t\in[T+1,n]$ we have $Z_t-Z_{t-1}=0$ by definition. Therefore, applying
Lemma~\ref{lem:prel2} gives that
\[
\E\left[e^{Z_n}\right] \leq e^{Z_0} = e^{4\|\theta\|_2^2} .
\]
Thus, combining everything together, we get
$
\E\left[e^{\disc_n}\right] \leq \E\left[e^{Z_n}\right] e^{16\|\theta\|_2^2} \leq e^{20\|\theta\|_2^2}$ as needed.
\end{proof}

\section{Applications}
\label{sec:applications}
In this section, we list some applications of our main result.
%
%
%
%
%
%
%
\subsection{$\ell_p$ variant of Koml\'{o}s problem}

\noindent{\bf Corollary~\ref{thm:lpdiscrepancy} (restated).} {\em
There is an efficient randomized algorithm which
given an $m\times n$ matrix $A$ having all columns of $\ell_2$ norm at most one and $p\in [1,\infty)$, finds a coloring $x\in\{-1,1\}^n$ with expected $\ell_p$ discrepancy $O(\sqrt{p})$.
}

\begin{proof}
Let $Y=Ax$.
By Theorem~\ref{thm:main}, we get that $Y$ is a $\sigma$-subgaussian random vector and hence each component of $Y$ is a $\sigma$-subgaussian random variable. Letting $Y_i$ denote the $i^{th}$ component of $Y$, we get
\[  \E[ \| Y\|_p^p] = \sum_{i=1}^m \E[|Y_i|^p] \le  mC^p p^{p/2}   \]
for a constant $C=C(\sigma)$ depending only on $\sigma$. The inequality above follows from the standard fact that the $p^{th}$ moment of an $O(1)$-subgaussian random variable is at most $C^pp^{p/2}$ (see e.g.~\cite{Roman10}).
\end{proof}

\subsection{Discrepancy relative to $\gamma_2$-norm}
\noindent{\bf Corollary~\ref{thm:herdisc} (restated).} {\em
There exists an efficient randomized algorithm that, given any $m\times n$ matrix $A$ and $J\subseteq[n]$, returns a coloring $x:J\rightarrow\{-1,1\}$ such that with constant probability,
\[ \| A_{|J} x\|_{\infty} = O(\gamma_2(A)\sqrt{\log m}) .\] }
\begin{proof}
It was shown in \cite{MNT14} that for any matrix $A$, there exists an efficiently computable quantity $\gamma_2(A)$ such that
\[ \Omega(\gamma_2(A)/\log m)\le herdisc(A) \le O(\gamma_2(A)\sqrt{\log m}) .\]
The upper bound above was proved using Theorem~\ref{thm:bana} (roughly) as follows: given any set $J\subseteq [n]$, factorize the matrix $A_{|J}$ as $A_{|J}=BC$
where every row of $B$ has $\ell_2$ norm at most $\gamma_2(A)$ and every column of $C$ has $\ell_2$ norm at most one. Such a factorization exists and can be computed by solving an appropriate semidefinite program. We refer the reader to \cite{MNT14} for more details.

Then using the matrix $B$, a convex body $K$ is defined as follows:
\[ K=\{ y\in\mathbb{R}^m: \| By\|_{\infty} \le c\gamma_2(A)\sqrt{\log m} \} . \]
For $c$ a large enough constant, $\gamma_m(K)\ge 1/2$. This follows by standard Gaussian tail bounds and union bound. Now finding a coloring $x\in\{-1,1\}^n$ of the columns of $A|_J$ such that the $\ell_\infty$ norm of $A|_J x=BCx$ is $O(\gamma_2(A)\sqrt{\log m})$ is equivalent to finding a coloring $x$ of the column vectors of $C$ such that $Cx$ lies in $K$. As $C$ has all columns of length at most one, by Theorem~\ref{thm:bana} there exists a coloring $x$ such that the  discrepancy of $A|_J$ is $O(\gamma_2(A)\sqrt{\log m})$.

Now Theorem~\ref{thm:main} (with Theorem~\ref{thm:subgaussian}) directly gives an efficient algorithm to find such a coloring.
\end{proof}

\subsection{A Generalization of Banaszczyk's Result}

%

In this subsection we prove the following generalization of Banaszczyk's result. The proof follows along similar lines as the proof of \cite{LSV} showing that linear discrepancy is at most twice the hereditary discrepancy.\\

\noindent{\bf Theorem~\ref{thm:genbana} (restated).} {\em
Let $S_1,S_2,\dots,S_n$  be sets such that for each $i\in [n]$, $S_i \subseteq \mathbb{B}_2^m$ and $0$ lies in the convex hull of $S_i$. Then for any convex body $K$ with $\gamma_m(K)\ge 1/2$, there exist vectors $v_i\in S_i$ such that
$\sum_{i=1}^n v_i \in cK$,
where $c>0$ is an absolute constant. Moreover, there is an efficient algorithm to find these vectors.
}
\begin{proof}
For technical reasons, we give the proof for the case when $\gamma_m(K)\ge 1/2+\epsilon$ for some $\epsilon>0$. The running time of the algorithm will be proportional to $\log(1/\epsilon)$. The general case, for instance when $\gamma_m(K)=1/2$, is slightly more complicated but can be handled by combining our proof with the techniques in \cite{DGLN16}. We provide a sketch of how to do this towards the end of the proof.

For each $i\in [n]$, as $0$ lies in the convex hull of $S_i$, there exist at most $m+1$ vectors in $S_i$ and a convex combination of them that equals $0$.
That is, there exist $m+1$ vectors $v_{i,j}\in S_i$ and real numbers $x_{i,j}\ge0$ such that
\begin{equation}
\label{eqn:genbana2}
\forall i \in [n]: \quad \sum_{j=1}^{m+1} x_{i,j}=1 \quad\text{ and }\quad\sum_{j=1}^{m+1} x_{i,j}v_{i,j}=0 .
\end{equation}
Our goal will be to round each collection $\{x_{i,j}: j \in [m+1]\}$ such that exactly one of them is $1$ and the rest are $0$, without incurring too much discrepancy.

Assume for now that each of $x_{i,j}$ can be expressed using at most $k$ digits in binary, for some finite $k$; that is, $x_{i,j} \in 2^{-k}\mathbb{Z}$ for all $i,j$.
We will see later how to get rid of this assumption. The main step will to be to reduce $k$ to $k-1$ by rounding the $k$-th bits in each $x_{i,j}$ either up or down,
in such a way that discrepancy stays bounded. Then, we will repeat this operation until we obtain $k=0$, which means that $x_{i,j} \in \{0,1\}$.
We will maintain the property that $x_{i,j} \ge 0$ and that $\sum_j x_{i,j}=1$ for all $i$, which then implies that there exists $j_{i} \in [m+1]$ for each $i$ such that $x_{i,j_i}=1$
and $x_{i,j}=0$ if $j \ne j_i$. The claim will follow by choosing $v_{i,j_i} \in S_i$.

Setting notation, let $x^{(k)}_{i,j} = x_{i,j}$. Assume that we already computed for $\ell \in [k]$ a choice of $x^{(\ell)}_{i,j} \in 2^{-\ell} \mathbb{Z}$ that satisfies
$\sum_j x^{(\ell)}_{i,j}=1$ for all $i \in [n]$. If $x^{(\ell)}_{i,j} \in \{0,1\}$ for all $j$ then there is nothing more to do.
Otherwise,
let $x^{(\ell)}_{i,j}=0.b^{(\ell)}_{i,j,1}b^{(\ell)}_{i,j,2}\dots b^{(\ell)}_{i,j,\ell}$ be the binary expansion of $x^{(\ell)}_{i,j}<1$.
Let $A^{(\ell)}=\{(i,j) \in [n] \times [m+1]: b^{(\ell)}_{i,j,\ell}=1\}$ be the indices where the $\ell$-th bit of $x_{i,j}$ is $1$.
We will construct a coloring $\chi^{(\ell)}:A^{(\ell)} \to \{-1,1\}$ which satisfies
\begin{equation}
\label{eqn:genbana}
 \sum_{(i,j) \in A^{(\ell)}} \chi^{(\ell)}(i,j) v_{i,j} \in c K \quad\text{ and }\quad \forall i\in[n], \; \sum_{j: (i,j) \in A^{(\ell)}}\chi^{(\ell)}(i,j)=0
\end{equation}
for some absolute constant $c>0$. Given such a coloring $\chi^{(\ell)}$, we compute the value of $x^{(\ell-1)}_{i,j}$ as follows:
\[
x^{(\ell-1)}_{i,j} =
\begin{cases}
x^{(\ell)}_{i,j} & \text{if } b_{i,j,\ell}^{(\ell)}=0\\
x^{(\ell)}_{i,j}+2^{-\ell} & \text{if } b_{i,j,\ell}^{(\ell)}=1 \text{ and } \chi^{(\ell)}(i,j)=1\\
x^{(\ell)}_{i,j}-2^{-\ell} & \text{if } b_{i,j,\ell}^{(\ell)}=1 \text{ and } \chi^{(\ell)}(i,j)=-1\\
\end{cases}\;.
\]
Observe that this operation zeros the $\ell$-th bit of all $x^{(\ell)}_{i,j}$, namely $x^{(\ell-1)}_{i,j} \in 2^{-(\ell-1)} \mathbb{Z}$;
it maintains the property that for all $i \in [n]$, $\sum_{j} x^{(\ell-1)}_{i,j}=1$ and $x_{i,j}^{(\ell-1)} \ge 0$;
and it satisfies that
$$
\sum_{i,j} \left( x^{(\ell-1)}_{i,j} - x^{(\ell)}_{i,j} \right) v_{i,j} \in 2^{-\ell} cK.
$$
Thus, repeating this rounding operation for $\ell=k,k-1,\ldots,1$ will result in a choice of $x_{i,j} \in \{0,1\}$ that satisfy $\sum_j x_{i,j}=1$ and
$$
\sum_{i,j} x_{i,j} v_{i,j} \in c K.
$$
That is, there is a choice of $j_i \in [m+1]$ for all $i \in [n]$ such that $\sum_{i,j} v_{i,j_i} \in c K$, as claimed.

It remains to show how to find a coloring satisfying \eqref{eqn:genbana}. Let $A^{(\ell)}_i=\{j: (i,j) \in A^{(\ell)}\}$ be the elements being colored in the set $S_i$ when we
round the $\ell$-th bits. We claim that $|A^{(\ell)}_i|$ must be even. This is since for every $i \in [n]$, the number of elements $x^{(\ell)}_{i,j}$
for which $b^{(\ell)}_{i,j,\ell}=1$ must be even, as $\sum x^{(\ell)}_{i,j}=1$. We will pair up the elements of $A^{(\ell)}_i$ in an arbitrary way, and only consider
colorings $\chi^{(\ell)}$ which give opposite colors to elements in each pair. In such a way, such a coloring automatically satisfies that $\sum_j \chi^{(\ell)}(i,j)=0$ for all $i \in [n]$.

For simplicity of notation, denote the vectors $\{v_{i,j}: j \in A^{(\ell)}_i\}$ by $\{u_{i,1},\ldots,u_{i,2 q_i}\}$ for some integer $q_i=|A^{(\ell)}_i|/2$. Define new vectors
$w_{i,j} = \left(u_{i,2j-1}-u_{i,2j}\right)/2$ for $j \in [q_i]$.
Observe that $\| w_{i,j}\|_2 \le 1$. Apply Theorem~\ref{thm:bana} to the vectors $w_{i,j}$ and the convex body $K$. This gives a coloring $\chi'(i,j)$
for each vector $w_{i,j}$ such that
\[ \sum_{i,j} \chi'(i,j) w_{i,j} \in cK.  \]
We now define the coloring $\chi^{(\ell)}$ to give the color $\chi'(i,j)$ to $u_{i,2j-1}$ and the color $-\chi'(i,j)$ to $u_{i,2j}$.
Clearly this satisfies both the conditions in \eqref{eqn:genbana} with constant $2c$.

We will now show that we can assume the binary expansion to be finite. Concretely, we will show that a preliminary rounding procedure
can allow us to assume that $k \le \log(2mn/\epsilon)$. This is since by truncating each $x_{i,j}$ after $\log (2mn/\epsilon)$ bits, the sum $\sum_{i,j} x_{i,j} v_{i,j}$ changes by at most
\[\sum_{i,j} \frac{1}{2^{\log (2mn/\epsilon)}}v_{i,j} \in \epsilon \mathbb{B}_2^m \subseteq K\]
and thus increases the final value of $c$ by at most $1$. The last containment follows by our assumption that $\gamma_m(K)\ge 1/2+\epsilon$, and thus $K$ must contain a Euclidean ball of radius $r$ which satisfies $\gamma_1([0,r]) \le \epsilon$. Clearly this is true for $r=\epsilon$.

We mention briefly now on how to tackle the case when $\gamma_m(K)<1/2+\epsilon$. This proceeds along similar lines as Theorem 40 from \cite{DGLN16}. The main idea is that we can find a point $p$ such that
\[p\in K\cap \left( \sum_{i} conv(S_i) \right)  \]
where $S_i$ denotes the convex hull of $S_i$ and the summation operator used is Minkowski addition. We then instead solve a new problem on the instance given by convex body $K':=\alpha(K-p)$ and sets $S_i'$ such that $\sum_{i} conv(S_i') =\alpha(\sum_{i} conv(S_i) -p)$ for some constant scaling factor $\alpha>0$. A solution of our original problem is recoverable from a solution of this. $p$ and $\alpha$ moreover satisfy the property that $\gamma_m(K') \ge 1/2+\epsilon$, and we already know how to solve this case.
\end{proof}

\section{Conclusion and Open Questions}
We gave efficient algorithms for several problems to find colorings with discrepancy bounds similar to those achievable using  Banaszczyk's result, Theorem~\ref{thm:bana}. However there are still some problems that use Banaszczyk's techniques in a non-trivial iterative way, for which we are unable to obtain an efficient algorithm.

One such problem is the Tusnady's problem about the discrepancy of axis-parallel boxes in $\mathbb{R}^d$. \cite{Sasho16} used Banaszczyk's technique to prove that the discrepancy is $O_d(\log^{d-1/2}n)$, where $O_d(.)$ hides factors depending only on $d$. Our techniques do not seem to apply here and the best known algorithmic bound is $O_d(\log^d n)$ \cite{BG17}.

Another such problem is the Steinitz problem in the $\ell_2$ norm. Here we are given $n$ vectors $v_1,\dots,v_n\in\mathbb{B}_2^m$ such that $\sum_i v_i=0$, and the goal is to find a rearrangement of these vectors such that the $\ell_2$ norm of the sum of vectors in any prefix along the rearrangement is small. That is, we want to find a permutation $\pi:[n]\rightarrow [n]$ to minimize
\[ \max_{k=1,\ldots,n} \left\| \sum_{i=1}^k  v_{\pi(i)}  \right\|_2 . \]
Banaszczyk in \cite{Bana12}, using techniques from \cite{B97}, showed that there exists a permutation for which the above expression is at most $O(\sqrt{m}+\sqrt{\log n})$, whereas the best known algorithmic bound is $O(\sqrt{m\log n})$ \cite{BG17}.


%
\bibliographystyle{alpha}
{\small \bibliography{refr} }

\end{document}